\documentclass[12pt]{article}

\usepackage{fullpage,comment,authblk,stackrel}
\usepackage[small]{caption}
\usepackage{tikz-cd}
\usepackage{nicefrac}
\usepackage[colorlinks=true,allcolors=blue]{hyperref}
\usepackage{amssymb,mathtools,amsthm,amsmath}
\usepackage[all,cmtip]{xy}
\usepackage{eulervm,xfrac,mathdots}
\usepackage{cite} 
\usepackage{graphicx}
\usepackage{xcolor}
\usepackage{enumitem}
\usepackage[capitalize,noabbrev]{cleveref}
\usepackage{xspace}

\usepackage[ruled,linesnumbered,noend]{algorithm2e}

\usepackage{lineno}

\newtheoremstyle{amit}
{7pt}
{7pt}
{}
{}
{\bf}
{:}
{.5em}
{}

\theoremstyle{amit}
\newtheorem{definition}{Definition}[section]
\newtheorem{prop}[definition]{Proposition}
\newtheorem{lemma}[definition]{Lemma}
\newtheorem{thm}[definition]{Theorem}
\newtheorem{corollary}[definition]{Corollary}
\newtheorem{rmk}[definition]{Remark}

\newcommand{\ssx}{\sigma}
\newcommand{\tsx}{\tau}

\newcommand{\lc}{\operatorname{lower}}
\newcommand{\uc}{\operatorname{upper}}

\usepackage{soul}

\usepackage{stackengine}

\newcommand{\Inc}{\mathsf{Inc}}
\newcommand{\Z}{\mathbb{Z}}

\newcommand{\Int}{\mathsf{Int}\,}
\newcommand{\Dgm}{\mathsf{Dgm}}
\newcommand{\field}{\mathsf{k}\,}
\newcommand{\rank}{\mathsf{rk}}
\renewcommand{\vec}{\mathsf{Vec}}

\newcommand{\birth}{\operatorname{birth}}

\title{Output-sensitive Computation of Generalized Persistence Diagrams for 2-filtrations
\thanks {
This work was supported by the Laboratory Directed Research and Development
Program (LDRD) of Lawrence Berkeley National Laboratory under U.S.\ Department
of Energy Contract No.\ DE-AC02-05CH11231 (DM).
This work is partially funded by Leverhulme Trust grant VP2-2021-008 (AP).
}}
\author[1]{Dmitriy Morozov}
\author[2]{Amit Patel}
\affil[1]{Lawrence Berkeley National Laboratory}
\affil[2]{Department of Mathematics, Colorado State University}
\date{}

\begin{document}
\maketitle

\begin{abstract}
    When persistence diagrams are formalized as the M\"obius inversion of the
    birth--death function, they naturally generalize to the multi-parameter
    setting and enjoy many of the key properties, such as stability, that we
    expect in applications.  The direct definition in the 2-parameter setting,
    and the corresponding brute-force algorithm to compute them, require
    $\Omega(n^4)$ operations, where $n$ is the complexity of the input.
    But the size of the generalized persistence
    diagram, $C$, can be as low as linear (and as high as cubic).  We elucidate
    a connection between the 2-parameter and the ordinary 1-parameter settings,
    which allows us to design an output-sensitive algorithm, whose running time
    is in $O(n^2 m + Cm)$, where $m$ is the number of simplices in the
    input.
\end{abstract}

\section{Introduction}
\label{sec:introduction}

An ordinary 1-parameter persistence diagram has a remarkable number of equivalent
definitions:
via persistent homology groups~\cite{ELZ02},
as in\-de\-com\-pos\-able summands of persistence
modules~\cite{ZoCa05,CadS10},
via well groups~\cite{well-groups-1D} or
persistence landscapes~\cite{Bub15},
and as the M\"obius inversion of
the rank invariant~\cite{Patel2018},
to name a few.
But extending
these to the multi-parameter setting leads to very different objects with wildly
different properties~\cite{CaZo09,well-groups,Vip20,kim2018generalized,McCPa20-edit-distance}.

The latter definition, although implicitly recognized in the
inclusion--exclusion formula used in the original proof of
stability~\cite{stability-persistence} and in size theory~\cite{size_theory},
received no attention until recently.
Recognizing this formula as a special case of the M\"obius inversion allows
for generalizations in many directions.
Kim and M\'emoli were the first to apply the M\"obius inversion to multiparameter persistence modules \cite{kim2018generalized}.
They define the persistence diagram as the M\"obius inversion of the generalized rank invariant.
We, however, are using a different definition.
For us, the persistence diagram is the M\"obius inversion of the
birth--death function~\cite{GregH1, GregH2}, specifically,
as defined by McCleary and Patel~\cite{McCPa20-edit-distance}
and further generalized by G\"ulem and McCleary~\cite{GulenMcCleary}.
The most important property of this approach is that it is functorial.
Functoriality plays a big role in justifying many of the choices we make in this paper.
It is important to note there are other approaches to persistence using the M\"obius inversion~\cite{signed_barcode, GradedPersistenceDiagrams, Asashiba}.
%

A question that remains open is how to efficiently compute the generalized
persistence diagram.
It is unclear how to take advantage of existing work.
Computing in\-de\-com\-pos\-able summands~\cite{DX19} doesn't give a generalized
persistence diagram, except in special cases for a different formulation~\cite{BLO20}.
Computing a
minimal presentation of a module~\cite{LW19,KR21} ought to help in general,
although not in the specific setting described in this paper.

We consider the case of the 2-filtration. We assume that the underlying
simplicial complex has $m$ simplices, and the sum of the number of minimal
grades at which the simplices appear --- the amount of information needed to
describe the filtration --- is $n \geq m$. ($n = m$ is the so-called 1-critical
case.)
It is possible to use the definition of the M\"obius inversion directly as an algorithm
(expressed below as \cref{cor:compact-formula}).
This formulation has $16 \cdot n^4$ terms,
and relies on the ability to compute the rank of any map between
a pair of homology groups in the 2-filtration.
An $O(m^4)$ algorithm for the latter task in the 1-critical setting is given, for
example, in \cite[Section 4.4.2]{morozov-phd}; it can be readily adapted to the
general setting.
Together the two observations lead to an $O(n^4)$ algorithm~\cite{BLO20}.
Unfortunately, because it has to examine all intervals in the 2-filtration, the
algorithm is also in $\Omega(n^4)$.

On the other hand, a generalized persistence diagram can be very sparse. Its
support, i.e., the number of non-zero intervals, can be as low as $m$, for
example, if the $2$-filtration is 1-critical and any monotone
path gives the same ordering of simplices. We use $C$ to denote the
number of non-zero intervals, the size of the output in our problem.

Our contributions are two-fold. After recapping the necessary background in
\cref{sec:background,sec:preliminaries}, we establish a connection, in
\cref{sec:mi-from-transpositions}, between the
non-zero intervals in the M\"obius inversion
and the pairing switches~\cite{vineyards} between simplices along four paths
through the 2-filtration. In \cref{sec:algorithm}, we develop an algorithm for
computing the generalized persistence diagram that traverses the 2-filtration
via a sequence of paths
by performing
transpositions of adjacent simplices.  It maintains extended
pairing information (that we call the \emph{birth curves}), which allows us to
compute all the intervals in the output-sensitive $O(n^2m + Cm)$ time.



\section{Background}
\label{sec:background}


\subsection{M\"obius inversion}
Let $P$ be any finite poset.
For every pair of grades $a, b \in P$, the \emph{interval} $[a,b]$ is the subset $\{ x \in P : a \leq x \leq b \}$.
The set of all intervals $\Int P$ is a poset, where $[a , b] \leq [c, d]$
whenever $a \leq c$ and $b \leq d$.
A \emph{monotone function} between posets is a function $f : P \to Q$ such
that for every $a \leq b$ in $P$, $f(a) \leq f(b)$ in $Q$.
The operator $\Int$ takes $f$ to a monotone function
$\Int f : \Int P \to \Int Q$ where $\Int f [a,b] := \big[ f(a), f(b) \big]$.

We are  interested in the $\Z$-incidence algebra on $\Int P$, denoted $\Inc ( \Int P)$.
It is the set of all integral functions
$\alpha : \Int \Int P \to \Z$ along with two binary operations:
    \begin{align*}
    (\alpha + \beta)\big(  [a,b], [c,d] \big) &:= \alpha \big(  [a,b], [c,d] \big) + \beta \big(  [a,b], [c,d] \big) \\
    (\alpha \ast \beta)\big(  [a,b], [c,d] \big) &:= \sum_{[a,b] \leq [x,y] \leq [c,d]} \alpha\big(  [a,b], [x,y] \big)
	\cdot \beta\big(  [x,y], [c,d] \big).
    \end{align*}
The additive identity is the \emph{zero function}, and the multiplicative identity
is the \emph{delta function} defined as $\delta\big(  [a,b], [c,d] \big) = 1$, for $[a,b] = [c,d]$, and $0$ otherwise.
One may think of elements of $\Inc (\Int P)$ as square matrices where both columns and rows
are indexed by $\Int P$.
Addition is matrix addition, and multiplication is matrix multiplication.

We are interested in two special functions in $\Inc( \Int P)$: the zeta function and the M\"obius function.
The \emph{zeta function} is the function $\zeta\big(  [a,b], [c,d] \big) = 1$, for all $[a,b] \leq [c,d]$,
and $0$ otherwise.
The zeta function is invertible and its (multiplicative) inverse, called
the \emph{M\"obius function} $\mu$, can be described inductively as follows
(see~\cite[Proposition 1]{Rot64}):
    \begin{equation}
    \label{eq:mo_function}
    \mu \big( [a,b], [c,d] \big) =
    \begin{cases}
    1 & \text{for } [a,b] = [c,d] \\
    - \sum_{ [a,b] < [x,y] \leq [c,d]} \mu\big(  [x,y], [c,d] \big) & \text{for } [a,b] < [c,d] \\
    0 & \text{otherwise.}
    \end{cases}
    \end{equation}

Given a function $f : \Int P \to \Z$, there is a unique function $g : \Int P \to \Z$ such that
	\begin{equation}
	\label{eq:mobius_inversion}
	f[c,d] = \sum_{[a,b] \leq [c,d]} g [a,b].
	\end{equation}
This unique function $g$ is called the \emph{M\"obius inversion} of $f$ (see~\cite[Proposition 2]{Rot64}).
It can be calculated as $f$ convolved with the M\"obius function $\mu$ as follows:
    \begin{equation}
    \label{eq:inversion}
    g[c,d] = \sum_{[a,b] \leq [c,d]} f[a,b] \cdot \mu \big( [a,b], [c,d] \big).
    \end{equation}
To recover $f$ from $g$, convolve with the zeta function as follows:
\begin{align*}
\sum_{[a,b] \leq [c,d]} g[c,d] &= \sum_{[a,b] \leq [c,d]} g[a,b] \cdot \zeta \big( [a,b], [c,d] \big) \\
&=
\sum_{[a,b] \leq [c,d]} \left ( \sum_{[x,y] \leq [a,b]} f[x,y] \cdot \mu \big([x,y], [a,b] \big) \right ) \cdot \zeta \big( [a,b], [c,d] \big) \\
&= \sum_{[x,y] \leq [c,d]} f[x,y] \cdot  \left( \sum_{[x,y] \leq [a,b] \leq [c,d]}
\mu \big([x,y], [a,b] \big) \cdot \zeta \big( [a,b], [c,d] \big) \right)  \\
& =
\sum_{[x,y] \leq [c,d]} f[a,b] \cdot \delta \big(  [x,y], [c,d] \big) = f[c,d].
\end{align*}
Thus, we interpret the convolution of $f$ with $\mu$ as the derivative of $f$ and the convolution of~$g$ with~$\zeta$ as the integral of $g$.

\subsection{Filtrations}
We now describe the (generalized) persistence diagram of a filtration as the M\"obius inversion
of its birth--death function as defined in \cite{McCPa20-edit-distance}.

Let $P$ be a finite poset with a unique maximal element, denoted $\top$, and a unique minimal element,
denoted $\bot$.
Fix a finite simplicial complex $K$, and let $\Delta K$ be the poset of all subcomplexes of $K$
ordered by inclusion.
	\begin{definition}
	A {\bf $P$-filtration} of $K$ is a monotone function $F : P \to \Delta K$ such that
	$F(\bot) = \emptyset$ and $F(\top) = K$.
	Note that when $P$ is totally ordered, a simplex enters the filtration at a unique element of $P$.
	For general $P$, a simplex may enter at multiple incomparable grades of~$P$.
	\end{definition}

For each dimension $d$, denote by $C_d(K)$ the $\field$-vector space
generated by the set of $d$-simplices in $K$.
For every $a \in P$, let $Z_d F(a) \subseteq C_d(K)$ be the subspace of $d$-cycles in $F(a)$ and
let $B_d F(a) \subseteq C_d(K)$ be the subspace of $d$-boundaries in $F(a)$.

\begin{definition}
The {\bf $d$-th birth--death function} of the $P$-filtration $F$ is the monotone integral function
$ZB_d F : \Int{P} \to \Z$ defined as follows:
	\begin{equation*}
	ZB_d F  [a,b] =
		\begin{cases}
		\dim \big( Z_d F(a) \cap B_d F(b) \big) & \text{ if  $b \neq \top$} \\
		\dim Z_d F(a) & \text{ if $b = \top$.}
		\end{cases}
	\end{equation*}
\end{definition}

The second case of $b = \top$ is necessary because it captures
cycles that live forever.
Without it, the persistence diagram of $F$, as defined below, will ignore all
essential cycles of~$K$.

\begin{definition}
The {\bf $d$-th persistence diagram} of the $P$-filtration $F$ is the M\"obius inversion, denoted
$\Dgm_d F$, of its the birth--death function $ZB_d F$.
\end{definition}

We will suppress the dimension $d$ in $Z_d F$, $B_d F$, $ZB_d F$, and $\Dgm_d F$
when it is of no importance.

\begin{rmk}
The above idea for defining the $d$-th persistence diagram of a $P$-filtration
also works for any $P$-module.
That is, let $M: P \to \vec$ be a functor, i.e., $P$-module.
Then there is a birth-death function
associated to $M$ and its persistence diagram is its M\"obius inversion~\cite{GulenMcCleary}.
Further, the persistence diagram $\Dgm_d F$ of a $P$-filtration $F$, as defined above, is the same, up to the diagonal,
as the persistence diagram of the $P$-module $H_d F$ obtained by applying homology to the filtration~\cite{Patel_Rask}.
\end{rmk}

\subsection{M\"obius Inversion for 1-Filtrations}
We now study the familiar case of $1$-filtrations, which are filtrations
over totally ordered posets.

Let $P_n$ be the totally ordered poset $\{ \bot = 0 < 1 < \cdots < n = \top \}$.
The lemma below follows from Equation (\ref{eq:mo_function})
by elementary calculations.
See \cite[Chapter 3.8]{10.5555/2124415} for a guide on how to compute
M\"obius functions.

\begin{lemma}
\label{lem:mobius}
The M\"obius function $\mu \in \Inc( \Int P_n)$ is particularly nice.
For every non-empty interval $[c,d] \in \Int P_n$,
    \begin{equation*}
    \mu \big( [a,b], [c,d] \big) =
    \begin{cases}
    (-1)^{ i} \cdot (-1)^{j} & \text{if  $\exists i,j \in \{0,1\}$ such that  $[a,b] = [c-i, d-j]$} \\
    0 & \text{otherwise}
    \end{cases}
    \end{equation*}
\end{lemma}

The following corollary is an immediate consequence of Equation (\ref{eq:inversion}) and Lemma \ref{lem:mobius}.

\begin{corollary}
\label{cor:compact-1D-formula}
Let $F : P_n \to \Delta K$ be a $1$-filtration and $\Dgm_d F$ its $d$-th persistence diagram.
Then, for any interval $[a,b] \in \Int P_n$,
\begin{align*}
    \Dgm_d F [a,b] & = \sum_{i,j \in \{0,1\}} (-1)^{i} \cdot (-1)^{j} \cdot ZB_d F[a-i, b-j].
\end{align*}
If the interval $[a-i, b-j]$ does not exist, then we interpret $ZB_d F[a-i, b-j]$ as zero.
\end{corollary}

\subsection{M\"obius Inversion for $2$-Filtrations}
We now consider $2$-filtrations, which are filtrations over the product of
two totally ordered posets and the main object of study in this paper.

Let $L_n := P_n \times P_n$ be the product poset defined as follows.
Its grades are tuples $a = (a_1, a_2)$, and the partial ordering is
$a = (a_1, a_2) \leq b = (b_1, b_2)$ whenever $a_1 \leq b_1$ and $a_2 \leq b_2$.
The minimal element is $\bot = (0,0)$ and the maximal element is $\top = (n,n)$.
The lemma below follows from Equation~(\ref{eq:mo_function})
by elementary calculations.
See \cite[Chapter 3.8]{10.5555/2124415} for a guide on how to compute
M\"obius functions.

\begin{lemma}
\label{lem:bimobius}
The M\"obius function $\mu \in \Inc(\Int L_n)$ is particularly nice.
For every non-empty interval $[c,d] \in \Int L_n$,
    \begin{equation*}
    \mu \big( [a,b], [c,d] \big) =
    \begin{cases}
    (-1)^{\# i} \cdot (-1)^{\#j} & \text{if } \exists i,j \in \{0,1\}^2 : [a,b] = [c-i, d-j] \\
    0 & \text{otherwise},
    \end{cases}
    \end{equation*}
   where $\# i := (i_1 + i_2)\! \mod 2$ and similarly for  $\# j$.
\end{lemma}

A \emph{$2$-filtration} is a filtration $F : L_n \to \Delta K$ indexed by the product of
two totally ordered posets.
The following corollary is an immediate consequence of Equation~(\ref{eq:inversion}) and
Lemma~\ref{lem:bimobius}.

\begin{corollary}
\label{cor:compact-formula}
Let $F : L_n \to \Delta K$ be a $2$-filtration and $\Dgm_d F$ its $d$-th persistence diagram.
Then, for any interval $[a,b] \in \Int L_n$,
\begin{align*}
    \Dgm_d F [a,b] & = \sum_{i,j \in \{0,1\}^2} (-1)^{\# i} \cdot (-1)^{\# j} \cdot ZB_d F[a-i, b-j].
\end{align*}
where $\# i = (i_1 + i_2) \mod 2$ and similarly for  $\# j$.
If the interval $[a-i, b-j]$ does not exist, then we interpret $ZB_d F[a-i, b-j]$ as zero.
\end{corollary}

\subsection{Rota's Galois Connection Theorem}
Functoriality of the (generalized) persistence diagram, as first observed in \cite{McCPa20-edit-distance},
plays an important role in justifying the seemingly arbitrary choices we make in this paper.
The recent work of G\"ulem and McCleary~\cite{GulenMcCleary} shows that functoriality of the persistence diagram is equivalent to Rota's Galois Connection Theorem.
We now state their version of this theorem for persistence diagrams.

A (monotone) \emph{Galois connection} from a poset $P$ to a poset $Q$, written
$f : P \leftrightarrows Q : g$,
are monotone functions $f$ and $g$ such that for all $a \in P$ and $x \in Q$,
$f(a) \leq x$ iff $a \leq g(x)$.

	\begin{thm}[\cite{GulenMcCleary}]
	\label{thm:galois}
	Let $F : P \to \Delta K$ and $G : Q \to \Delta K$ be two filtrations
	and $f : P \leftrightarrows Q: g$ a Galois connection.
	If $G = F \circ g$, then for every $x \leq y$ in $Q$,
		$$\Dgm_d G [x,y] = \sum_{[a,b] \in (\Int f)^{-1}[x,y]} \Dgm_d F[a,b].$$
	\end{thm}

This theorem says that if two filtrations are related by a Galois connection, then the persistence diagram of one can be read from the other.
\cref{sec:two_to_one,sec:coarsening-refining,sec:bd-rank} make use of this theorem.

\subsection{From 2-filtration to 1-filtrations}
\label{sec:two_to_one}
A 2-filtration over $L_n$ induces a family of 1-filtrations when restricted to paths in $L_n$.
We now apply Theorem~\ref{thm:galois} to study how their persistence diagrams relate.

	\begin{definition}
	A {\bf step} in $L_n$ is a pair of adjacent grades, i.e., grades that
	differ by 1 in a single position: $(a,b) \to (a+1,b)$, or $(a,b) \to (a,b+1)$.
	A {\bf path} in $L_n$ is a monotone
	sequence of $2n$ steps starting at~$(0,0)$ and ending at~$(n,n)$ .
	In other words, a path is a monotone, injective function $p : P_{2n+1} \to L_n$
	\end{definition}

The path $p$ induces a Galois connection $f : L_n \leftrightarrows P_{2n+1} : p$ as follows.
For any $a \in L_n$, define $f(a)$ as the unique minimal element $x \in P_{2n+1}$
such that $a \leq p(x)$.
Note $f(a) \leq x$ iff $a \leq p(x)$.
The following corollary follows from Theorem \ref{thm:galois}.

\begin{corollary}  \label{cor:paths}
Let $F : L_n \to \Delta K$ be a $2$-filtration and $p : P_{2n+1} \to L_n$ a path.
If $G$ is the $1$-filtration $F \circ p : P_{2n+1} \to \Delta K$,
then for any $x \leq y$ in $P_{2n+1}$,
	$$\Dgm_d G [x,y] = \sum_{[a,b] \in (\Int f)^{-1}[x,y]} \Dgm_d  F[a,b].$$
Here $f$ comes from the Galois connection $f : L_n \leftrightarrows P_{2n+1} : p$
described above.
\end{corollary}

In other words, it is possible to read off the ordinary persistence diagram for
any 1-filtration defined by a path through the poset $L_n$ from the generalized
persistence diagram of the 2-filtration.
In this paper, we develop an algorithm that does the opposite.
We glue the persistence diagrams for a family of paths into
the persistence diagram of the 2-filtration.

\subsection{Coarsening and refining}
\label{sec:coarsening-refining}
We now apply Theorem~\ref{thm:galois} to study how the persistence
diagram of a filtration changes under refinement.

Refine $P_n$ by dividing every step into multiple steps:
$$P_{n'} := \left \{ 0 < \ldots < 1 < \cdots < n-1 < \ldots < n \right \}.$$
The product $L_{n'} := P_{n'} \times P_{n'}$ is a refinement of $L_n$, and the
two are related  by a Galois connection $f : L_{n'} \leftrightarrows L_n : g$
as follows.
For any $(a, b) \in L_{n'}$, $f(a, b) := \big( \lceil a \rceil, \lceil b \rceil \big)$, and $g$ is inclusion.
The following corollary follows from Theorem \ref{thm:galois}.

\begin{corollary}
\label{cor:coarsening-refining}
Let $F : L_{n'} \to \Delta K$ and $G : L_n \to \Delta K$ be filtrations and $f : L_{n'} \leftrightarrows L_n : g$
the Galois connection described above.
If $G = F \circ g$, then for any $x \leq y$ in $G$,
	$$\Dgm_d G [x,y] = \sum_{[a,b] \in (\Int f)^{-1}[x,y]} \Dgm_d  F[a,b].$$
\end{corollary}

This corollary highlights a significant advantage of generalized persistence
diagrams: they are stable under refinement. Lack of this property is one of the
chief problems with using indecomposable summands to generalize persistence
diagrams (breaking ties in the input can significantly change the
decomposition).
Our definition of a 2-filtration allows two simplices to enter at the same grade.
However, our algorithm does not.
In this case, Corollary~\ref{cor:coarsening-refining} enables us to break ties arbitrarily;
see Section~\ref{sec:preliminaries}.

\subsection{Birth--death vs rank functions}
\label{sec:bd-rank}
We now explain our preference for the birth--death function over the better-known rank function
associated to a filtration.

Consider the classical definition of the persistence diagram associated
to a $1$-filtration $G : P_n \to \Delta K$.
The \emph{$d$-th rank function} associated to $G$
is the function $\rank_d : \Int P_n \to \Z$ that assigns to every interval $[a,b]$
the rank of the linear map on homology
$H_d\big( G(a)  \big) \to H_d \big( G(b) \big)$
induced by the inclusion $G(a) \hookrightarrow G(b)$.
The $d$-th persistence diagram of $F$ is defined as the assignment to every
$a < b$  (up to a shift), the integer
$$\rank_d [a,b] - \rank_d [a-1, b] + \rank_d [a-1,b+1] - \rank_d [a, b+1].$$
This assignment can be interpreted as the M\"obius inversion of $\rank_d F$ using the
containment~$\supseteq$ relation on $\Int P_n$ instead of the product ordering we are using
in this paper.
That is, $[c,d] \supseteq [a,b]$ if $c \leq a$ and $b \leq d$.

We prefer the birth--death function over the rank function because it behaves well with Galois connections.
Consider a Galois connection between arbitrary posets $f : P \leftrightarrows Q : g$ and
apply $\Int$. Assuming the product ordering on the two interval posets $\Int P$ and $\Int Q$,
we see
	\begin{align*}
	\Int f [a,b] \leq [x,y] \text{ iff } [a,b] \leq \Int g [x,y]
	\Leftrightarrow f(a) \leq x \text { and } f(b) \leq y
	\text{ iff } a \leq g(x) \text{ and } b \leq g(y).
	\end{align*}
In other words, $\Int f : \Int P \leftrightarrows \Int Q : \Int g$ is a Galois connection.
This observation is at the heart of the proof of Theorem \ref{thm:galois}.
Unfortunately, $\Int f : \Int P \leftrightarrows \Int Q : \Int g$ fails to be a Galois connection
under the containment ordering on $\Int  P$ and $\Int Q$.
This makes Theorem~\ref{thm:galois} unavailable when using the rank function.
We need Corollaries \ref{cor:paths} and \ref{cor:coarsening-refining} to justify the choices we make later in this paper.

%

We prefer the birth--death function both for its
theoretical properties and for the simpler interpretation of its M\"obius
inversion\footnote{
    Suppose $F$ is a 1-filtration, with $i$ the minimal grade of simplex
    $\ssx_i$. Define a 2-filtration $G$ so that the unique minimal grade of $\ssx_i$ is
    $(i,i)$. The M\"obius inversion of the birth--death function of $G$ is
    $\Dgm G[(i,i),(j,j)] = +1$ if $\Dgm F[i,j] = +1$ and $0$ otherwise.
    The M\"obius inversion of the rank function on the other hand has values,
      $\Dgm^{\rank} G[(i,i),(j,n)] = +1$,
      $\Dgm^{\rank} G[(i,i),(n,j)] = +1$,
      $\Dgm^{\rank} G[(i,i),(j,j)] = -1$
    for every $\Dgm F[i,j] = +1$.}.
But we also note that because the two produce the same persistence diagram for the
1-filtrations, the algorithm described in this paper can be immediately adapted to
computing the M\"obius inversion of the rank function on the 2-filtration; only
the output values have to be changed.

\subsection{Transpositions}
Fix a $1$-filtration $F : P_n \to \Delta K$ and
assume that for every adjacent pair of subcomplexes, the difference
$F(i) - F(i-1)$ is empty or a single simplex, denoted $\ssx_i$.
Given the boundary matrix $D$ of $K$, with rows and
columns ordered by the $1$-filtration, the standard persistence
algorithm~\cite{ELZ02} finds a
factorization~\cite{vineyards}, $R = DV$, where $R$ is \emph{reduced}, meaning
the lowest non-zero entries in its columns appear in unique rows, and $V$ is
invertible upper-triangular. The lowest non-zeros in matrix $R$ give the
persistence pairing: for $j \neq n$, $\Dgm [i,j] = 1$ iff $R[i,j] \neq 0$ and
$R[i',j] = 0 ~\forall~ i' > i$; for $j = n$, $\Dgm [i,n]$ is the number of
zero columns $i$ (such that $R[\cdot, i]$ is $0$) minus the number of non-zero columns $j$
(such that $\exists i, R[i,j] \neq 0$).

\begin{definition}
Whenever $\Dgm [i,j] \neq 0$, we say $\sigma_i$ is {\bf paired} with $\sigma_j$
and that $\sigma_i$ is {\bf positive} and $\sigma_j$ is {\bf negative}.
Two pairs $(\sigma_i, \sigma_j)$ and $(\tau_k, \tau_l)$ are {\bf nested}
if the interval $[i,j]$ is contained in the interval $[k,l]$.
They are {\bf disjoint} if $[i,j] \cap [k,l] = \emptyset$.
\end{definition}

Cohen-Steiner et al.~\cite{vineyards} (see also \cite{morozov-phd}) study what happens to the pairing when we
transpose two simplices in the 1-filtration $F$. They analyze how the decomposition
$R = DV$ may fail to satisfy the requirement that $R$ is reduced and $V$ is
invertible upper-triangular, and show that this property can be restored,
following a single transposition, in linear time.
\cref{sec:transposition-analysis} briefly recaps the details of the updates.
The following lemma is a consequence of their analysis.
\begin{lemma}[Nested-Disjoint Lemma~\cite{vineyards,morozov-phd}]
    \label{lem:nested-disjoint}
    The pairing of two transposing simplices can switch only if before the
    transposition their pairing is either nested, or disjoint. (If the switch
    occurs, the pairing remains nested or disjoint after the transposition.)
\end{lemma}
The contrapositive of this statement is an important shortcut that we use below:
if the pairing of two transposing simplices is neither nested, nor disjoint, it
will not change after the transposition.

\section{Preliminaries}
\label{sec:preliminaries}

We now develop notation specializing to $2$-filtrations and state an important
lemma.

An \emph{antichain} in $L_n$ is any subset such that no two grades are comparable.
An \emph{upset} is a subset $U \subseteq L_n$ such that if $a \in U$ and $a \leq b$, then $b \in U$.
Since $L_n$ is finite, its set of antichains is in bijection with its set of upsets.
Fix an upset $U$.
A \emph{lower corner} of $U$ is any minimal grade of $U$.
Note that the set of all lower corners of $U$ is an antichain;
we denote it $\lc(U)$.
We say $(i,j) \in U$ is an \emph{upper corner} of $U$
if $(i-1, j)$ and $(i,j-1)$ are in $U$ but $(i-1, j-1)$ is not in $U$.
We denote the set of all upper corners as $\uc(U)$.
A \emph{periphery} of $U$ is the set of grades $(i,j)$ such that grades
$(i-1,j-1)$ are not in the upset. See \cref{fig:upset}.
(The periphery uniquely determines the upset, so we abuse the terminology and
also refer to lower and upper corners of a periphery.)

\begin{figure}
    \centering
    \includegraphics{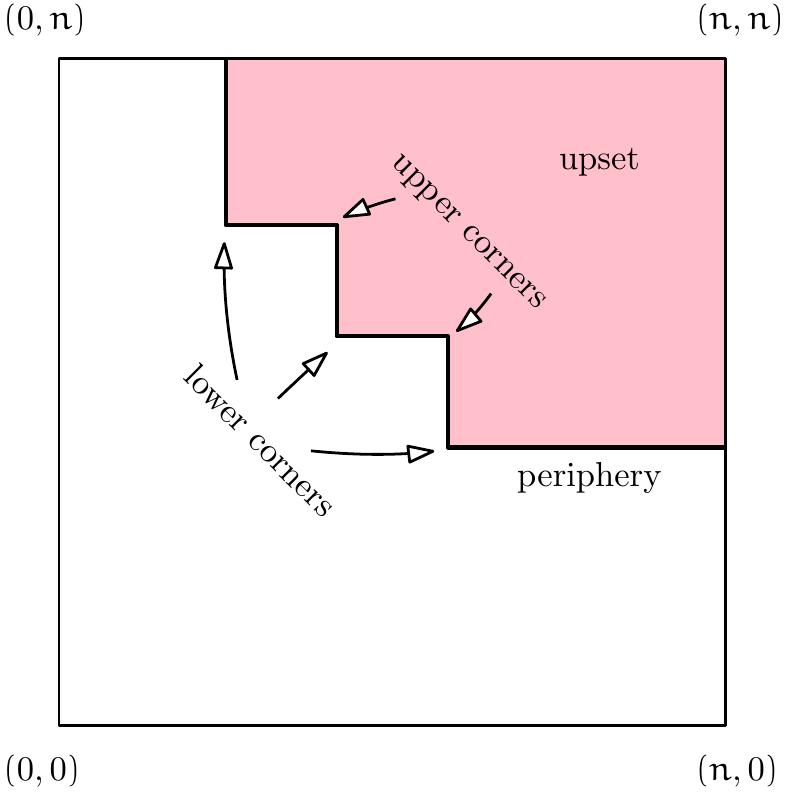}
    \caption{An upset in pink, its lower and upper corners, and periphery in bold.}
    \label{fig:upset}
\end{figure}

Fix a 2-filtration $F : L_n \to \Delta K$.
For any simplex $\sigma \in K$, let  $f(\ssx)$ denote the
periphery of the upset $U_\ssx$ that contains $\ssx$.
We call this the \emph{appearance curve} of $\ssx$.
Note that a path in $L_n$ enters upset $U_\ssx$
at a unique grade of the appearance curve.

\begin{definition}
The 2-filtration $F$ is {\bf  non-degenerate} if no two lower corners of any
appearance curves share a component of a grade.
That is, if $(a_1, a_2)$ is a lower corner of $f(\ssx)$ and
$(b_1, b_2)$ is a lower corner of $f(\tsx)$, then $a_1 \neq b_1$ and $a_2 \neq b_2$.
\end{definition}

Every degenerate filtration $F$ can be refined into a non-degenerate filtration $G$
as follows.
Consider a refinement $L_{n'}$ of $L_n$ and recall the Galois connection
$f : L_{n'} \leftrightarrows L_n : g$ from \cref{sec:coarsening-refining}.
Assuming $L_{n'}$ is sufficiently fine, there is a non-degenerate filtration $G : L_{n'} \to \Delta K$
such that $G = F \circ g$, with ties in $F$ broken arbitrarily, respecting
the dimension of the simplices.
\cref{cor:coarsening-refining} provides the equation for recovering $\Dgm F$ from $\Dgm G$.

Throughout the paper, we use $m = |K|$ to denote the size of the simplicial
complex, and $n = \sum_\ssx |\lc(f(\ssx))|$, the complexity of the input, i.e., the
size of the description of the minimal grades at which the simplices appear.

\begin{lemma}[Path Invariance]
    \label{lem:path-invariance}
    Let $F : L_n \to \Delta K$ be a non-degenerate filtration and $p : P_{2n+1} \to L_n$
    any path.
    If simplex $\ssx$ is added to the 1-filtration $F \circ p$ at step $(i,j) \to
    (i+\delta_i,j + \delta_j)$, simplex $\tsx$ is added to the 1-filtration at
    step $(k,l) \to (k + \delta_k, l + \delta_l)$, and $\ssx$ and $\tsx$ are
    paired in the 1-filtration, then they are paired in every 1-filtration induced by
   any path taking these two steps.
\end{lemma}
\begin{proof}
    Let $K_1$ denote the complex $F (i,j)$, $K_2 = K_1 \cup \{ \ssx \}$
    denote the complex $F(i+\delta_i, j + \delta_j)$,~$K_3$ denote the complex $F(k,l)$,
    $K_4 = K_3 \cup \{ \tsx \}$ denote the complex $F(k + \delta_k, l + \delta_l)$.
    Then the 1-filtration
    induced by any path taking these two steps
    contains the following 2-step filtration as a subfiltration:
    $K_1 \subseteq K_2 \subseteq K_3 \subseteq K_4$.
    From \cref{cor:compact-1D-formula},
    $\ssx$ and $\tsx$ are paired iff
    \[
        \dim (Z K_2 \cap B K_4)
            - \dim (Z K_1 \cap B K_4)
            - \dim (Z K_2 \cap B K_3)
            + \dim (Z K_1 \cap B K_3) = 1.
    \]
    In other words, the pairing is independent of the order of
    simplices in $K_1$ and $K_3 - K_2$.
\end{proof}

\section{M\"obius Inversion from Transpositions}
\label{sec:mi-from-transpositions}

\begin{figure}
\begin{center}
\includegraphics{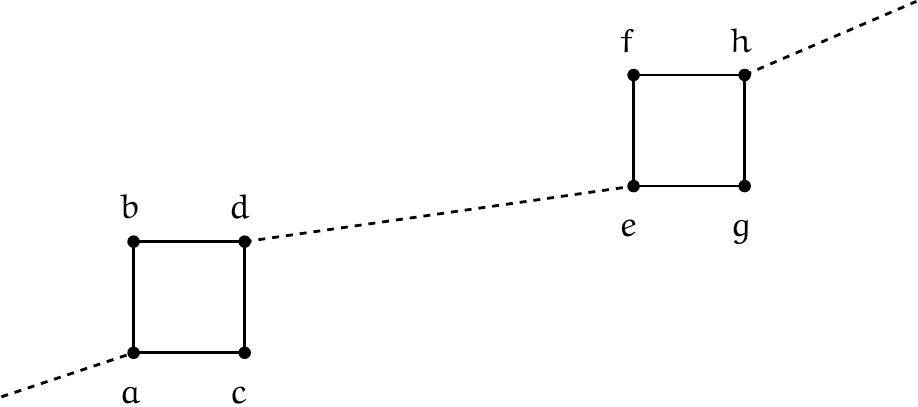}
\caption{1-dimensional paths through the 2-filtration involved in the analysis.}
\label{fig:four-paths-compact}
\end{center}
\end{figure}

Fix a non-degenerate filtration $F: L_n \to \Delta K$.
To compute the persistence diagram $\Dgm F$ at the interval $[d,h]$,
it suffices, by \cref{cor:compact-formula},
to consider any four 1-dimensional paths through the 2-filtration that go through the subcomplexes
$F(d-i)$ and $F(h-j)$ for $i,j \in \{0,1\}^2$. We denote the four grades around $d$ as $a,b,c,d$
and the four grades around $h$ as $e,f,g,h$; see \cref{fig:four-paths-compact}.
The four paths in question are:

\begin{equation}
\begin{tabular}{ccc}
    $\ldots abd \ldots efh \ldots$      & $\leftrightarrow$ &     $\ldots abd \ldots egh \ldots$  \\
    $\updownarrow$                      &                   &     $\updownarrow$                 \\
    $\ldots acd \ldots efh \ldots$      & $\leftrightarrow$ &     $\ldots acd \ldots egh \ldots$
\end{tabular}
\label{eq:four-paths}
\end{equation}

Because the filtration is non-degenerate, at most a single simplex appears at
every one of the eight steps on these paths. Moreover, the only way two
different simplices $\ssx$ and $\tsx$ can appear going from $a$ to $d$ (or from
$e$ to $h$) is if $d$ is a minimum of the intersection of upsets of the
appearance curves of $\ssx$ and $\tsx$.
In this case, going between any two paths connected by an arrow results in a
transposition of the two simplices in the 1-dimensional filtrations induced by
the paths.

\paragraph{2D from 1D.}
For any two nested spaces $F(x) \subseteq F(y)$, let $xy$ denote the value of the birth--death function $ZB F[x,y]$.
For any filtration
\[
    F(w) \subseteq F(x) \subseteq F(y) \subseteq F(z),
\]
let $wxyz$ be the value of the 1-dimensional M\"obius inversion of the birth--death
function for the interval $[x,z]$ in this filtration, namely
\begin{equation}
    wxyz = xz - xy - wz + wy. \label{eq:wxyz}
\end{equation}

With this notation, the 2-dimensional M\"obius inversion can be rewritten in
terms of the 1-dimensional M\"obius inversions:

\begin{lemma}
    \label{lem:dgm-from-paths}
    \begin{align}
        \Dgm F [d,h]
                &= (cdfh - abfh) - (cdeg - abeg)    \label{eq:vh-path}    \\
                &= (cdgh - abgh) - (cdef - abef)    \label{eq:vv-path}    \\
                &= (bdfh - acfh) - (bdeg - aceg)    \label{eq:hh-path}    \\
                &= (bdgh - acgh) - (bdef - acef)    \label{eq:hv-path}    \\
                &= adeh - adef - adeg - abeh - aceh + abef + abeg + acef + aceg.  \label{eq:9-paths}
    \end{align}
\end{lemma}
\begin{proof}
    Expanding any one of the five formulae using \cref{eq:wxyz} results in
    \[
        (dh - df - dg + de) - (ch - cf - cg + ce) -
            (bh - bf - bg + be) + (ah - af - ag + ae),
    \]
    which equals $\Dgm F [d,h]$ by \cref{cor:compact-formula}.
\end{proof}

If two different simplices, $\ssx$
and $\tsx$, appear as we go from $a$ to $d$ (or from $e$ to $h$) and their
pairing does not switch, after a transposition, along any path through the
bifiltration, then the value of the diagram $\Dgm F[d, \cdot]$
(or $\Dgm F[\cdot, h]$) is 0.

\begin{lemma}
    \label{lem:pairing-switch}
    Suppose $F(x - (1,1)) - F(x) = \{ \ssx, \tsx \}$.
    Let $p_1$ and $p_2$ be two paths that both pass through grades $x - (1,1)$
    and $x$ and only differ in the step they take in between those grades.
    If for any such two paths, the pairing of simplices $\ssx$ and $\tsx$ does
    not switch, then
    \[
        \Dgm F [x, \cdot] = 0 \qquad \textrm{and} \qquad \Dgm F [\cdot, x] = 0.
    \]
\end{lemma}
\begin{proof}
    We prove the first claim, $\Dgm F [x, \cdot] = 0$; the second claim is
    proved analogously.
    Suppose $x = d = (d_1, d_2)$ and assume, without loss of generality,
    that the appearance curve of $\ssx$ has a lower-corner $(\cdot, d_2)$,
    and the appearance curve of $\tsx$ has a lower-corner $(d_1, \cdot)$.
    From the 1-dimensional case,
    we know that if for any two paths going through $a$ and $d$,
    simplex $\ssx$ has the same pairing, then for any step $x \to y$, we
    have:
    \[
        abxy = cdxy \qquad \textrm{and} \qquad xyab = xycd,
    \]
    depending whether $d \leq x$ or $y \leq a$.

    Taking $x,y$ to be either $f,h$ or $e,g$, we get:
    \begin{align*}
        0 &= (cdfh - abfh) - (cdeg - abeg) \\
          &= \Dgm F [d, h].     \tag{By \cref{eq:vh-path} in \cref{lem:dgm-from-paths}.}
    \end{align*}
\end{proof}

\paragraph{Three cases.}
Because $F$ is non-degenerate, there are at most two persistence pairs between
$d$ and $h$ in the filtration
$F(a) \subseteq F(d) \subseteq F(e) \subseteq F(h)$.

\begin{lemma}
    If $F$ is non-degenerate, then $adeh \in \{0, 1, 2 \}$.
\end{lemma}
\begin{proof}
    Because the filtration is non-degenerate, at most one simplex appears at any
    step around each one of the squares, and therefore at most two simplices
    appear between $F(a)$ and $F(d)$ as well as between $F(e)$ and $F(h)$. It
    follows that at most two classes can be born between $a$ and $d$ and die
    between $e$ and $h$.
\end{proof}

We now analyze the three possible values of $adeh$ separately.

\subsection{No Pairs}
If there are no 1-dimensional persistence pairs between $d$ and $h$, then
the value of the 2-dimensional diagram is zero.

\begin{thm}
    If $adeh = 0$, then $\Dgm F[d,h] = 0$.
\end{thm}
\begin{proof}
    If $adeh=0$, then every one of the terms on the right-hand side of
    \cref{eq:9-paths} in
    \cref{lem:dgm-from-paths} are also zero. (If any one of them was positive, it would
    contribute to a pair between $d$ and $h$ in the coarser filtration.)
\end{proof}

\subsection{Single Pair}
\label{sec:single-pair}
If there is one 1-dimensional pair between $d$ and $h$, then the value of the
2-dimensional diagram has a simple expression. (This is a crucial result for our
algorithm; we spell out its implications after the proof.)

\begin{thm}
    \label{thm:product-formula}
    If $adeh = 1$, then the value of the 2-dimensional diagram is
    given by
    \[
        \Dgm F[d,h] = B \cdot D,
    \]
    where
    \begin{align*}
        B &= 1 - abeh - aceh \\
        D &= 1 - adef - adeg.
    \end{align*}
\end{thm}

To prove this theorem, we need a set of auxiliary formulae.

\begin{prop}
    \label{prop:path-product}
    If $adeh = 1$, then
    \begin{align*}
        abeh \cdot adef &= abef \\
        aceh \cdot adef &= acef \\
        abeh \cdot adeg &= abeg \\
        aceh \cdot adeg &= aceg. \\
    \end{align*}
\end{prop}
\begin{proof}
    We prove the first statement by considering the 1-dimensional filtration
    \[
        F(a) \subseteq F(b) \subseteq F(d) \subseteq F(e) \subseteq F(f) \subseteq F(h).
    \]
    The right-hand side of the first statement is 1 iff a class born at $b$ dies
    at $f$. Under the assumption $adeh=1$, consider the class that is born
    between $a$ and $d$ and dies between $e$ and $h$:
    \begin{enumerate}
        \item this class is born at $b$ iff $abeh = 1$ (otherwise, it is born at $d$).
        \item this class dies at $f$ iff $adef = 1$ (otherwise, it dies at $h$).
    \end{enumerate}
    It follows that if $adeh=1$, then $abeh \cdot adef = abef$.

    The other three statements are proved analogously.
\end{proof}

\begin{proof}[Proof of \cref{thm:product-formula}]
    \begin{align*}
        B  \cdot D & =  (1 - abeh - aceh)  \cdot (1 - adef - adeg) \\
                   & =   1 - adef - adeg - abeh - aceh\ + \\
                   &    \hspace{1in} abeh \cdot adef + abeh \cdot adeg + aceh \cdot
                     adef + aceh \cdot adeg \\
                   & =  1 - adef - adeg - abeh - aceh\ + \\
                   &    \hspace{1in} abef + abeg + acef + aceg   \tag{By \cref{prop:path-product}.} \\
                   &  =  adeh - adef - adeg - abeh - aceh\ +  \tag{By assumption $adeh=1$.}\\
                   &        \hspace{1in} abef + abeg + acef + aceg \\
                   & =  \Dgm F [d,h].     \tag{By \cref{eq:9-paths} in \cref{lem:dgm-from-paths}.}
    \end{align*}
\end{proof}

\paragraph{Implications.}
To understand the algorithmic implications of \cref{thm:product-formula}, we
need to dissect the terms $B$ and $D$. Because the terms on the right-hand side
of $B = 1- abeh - aceh$, involve only a single step ($a \to b$ or $a \to c$) through
the non-degenerate bifiltration, their values can be only 0 or 1. Accordingly,
$B$ can take on values $-1, 0, 1$.
\begin{itemize}
    \item
        $B = -1$ iff both $abeh$ and $aceh$ are 1.  This happens iff the class
        that dies between $e$ and $h$ is born along the steps $a \to b$ and $a
        \to c$ around the first square. If the same simplex appears at both
        steps, it means that $d$ is the upper-corner of its appearance curve.
        If two different simplices appear, it means that the pairing switches
        when we transpose them in the filtration and that the first-to-appear
        simplex creates the class that dies between $e$ and $h$.
    \item
        $B = 0$ iff only one of the two terms, $abeh$ or $aceh$, is 1. Without
        loss of generality, suppose $abeh = 1$. Then, because a class
        that dies between $e$ and $h$ has to be born between $a$ and $d$,
        $cdeh = 1$. This cannot happen if the same simplex appears at steps $a
        \to b$ and $a \to c$ (otherwise, we'd get a contradiction with
        \cref{lem:path-invariance}). If two different simplices appear at those steps,
        then the pairing does not switch when we transpose them.
    \item
        $B = 1$ iff both of the two terms, $abeh$ and $aceh$ are 0. In this
        case, the class that dies between $e$ and $h$ must be born along the
        steps $b \to d$ and $c \to d$ (which means $bdeh = cdeh = 1$).
        This is possible if $d$ is the lower corner of the appearance curve of a
        simplex. It is also possible for two different simplices to appear at
        the two steps (this happens if $d$ is a minimum of the intersection of
        upsets of their appearance curves).
        In this case, the pairing switches after their transposition,
        and the simplex that appears last is the one that
        creates the dying class.
\end{itemize}

Similar analysis applies to the term $D = 1 - adef - adeg$.
\begin{itemize}
    \item
        $D = -1$ iff the class born between $a$ and $d$ dies at the first step
        ($e \to f$ and $e \to g$) around the second square. This happens if the
        same simplex appears at the two steps, and $h$ is the upper-corner of
        its appearance curve; or if two different simplices appear, the pairing
        switches when their order is transposed, with the first-to-appear
        simplex involved in the pair.
    \item
        $D = 0$ iff the class dies at the first step or at the second step,
        depending on which way we go around the square. This happens if the
        pairing does not switch as we change how we go around the square.
    \item
        $D = 1$ iff the class dies at the second step ($f \to h$ and $g \to h$)
        around the second square. This happens if $h$ is a lower-corner in the
        appearance curve of a simplex, or if two different simplices appear
        and their pairing switches when they are transposed, with the
        last-to-appear simplex involved in the pair.
\end{itemize}

What is crucial for the algorithm in \cref{sec:algorithm} is that for the terms $B$
or $D$ to be non-zero, the respective squares must go around either the lower-,
or the upper-corner of the appearance curve of the same simplex, or if two
different simplices appear in the square, their pairing must switch, when we
transpose them in the filtration. (The latter case is consistent with
\cref{lem:pairing-switch}, but more precise. It tells us not only that some
value in the diagram might be non-zero, but also narrows down the locations of
the generalized pairs.)

\subsection{Two Pairs}
\label{sec:two-pairs}
If $adeh = 2$, then two pairs of simplices, $\ssx_1, \ssx_2$ and $\tsx_1,
\tsx_2$ appear in the first and in the second squares.

Consider the four paths around the two squares, see (\ref{eq:four-paths}), and
consider what happens to the pairing of the four simplices.
\begin{itemize}
    \item
        If the pairing never switches, then $\Dgm F [d,h] = 0$ by
        \cref{lem:pairing-switch}.
    \item
        If the pairing switches at least once, but doesn't switch everywhere,
        then along one of the paths the
        pairing has to be neither nested, nor disjoint (i.e., first simplex of
        $\ssx_1, \ssx_2$ is paired with the first simplex of $\tsx_1,\tsx_2$,
        and second is paired with second). Without loss of generality, suppose
        we get such a pairing along the path $\ldots abd \ldots efh \ldots$.
        \cref{lem:nested-disjoint} implies that no pairing switch is possible as
        we move to paths $\ldots acd \ldots efh \ldots$ and
        $\ldots abd \ldots egh \ldots$. It follows that the two pairings of the
        simplices along those two paths is nested. Since we assumed that at
        least one switch in pairing occurs, it follows that the pairing switches
        as we take either one of the two transitions to path
        $\ldots acd \ldots egh \ldots$, and in particular the pairs are nested.
        Using these four facts,
        implies
        \begin{align*}
            bdfh &= 1 \\
            acfh &= 1 \\
            bdeg &= 1 \\
            aceg &= 0.
        \end{align*}
        Substituting the four terms into \cref{eq:hh-path}, we get $\Dgm F[d,h] = -1$.
    \item
        If the pairing switches on each one of the four transpositions, then from
        \cref{lem:nested-disjoint}, it follows that the first simplex to appear
        of $\ssx_1$ and $\ssx_2$ is paired with the last simplex to appear of
        $\tsx_1$ and $\tsx_2$ (and vice versa, last-to-appear is paired with the
        first-to-appear). In this case, we can use any one of the
        \cref{eq:vh-path,eq:vv-path,eq:hh-path,eq:hv-path} to get
        $\Dgm F[d,h] = -2$.
        We note that this case occurs generically (more on that in the next
        section).
\end{itemize}



\section{Algorithm}
\label{sec:algorithm}

We use the observations in the previous section to devise an algorithm that
tracks the changes in pairing along 1-dimensional paths through the 2-filtration
and identifies all intervals in the support of the generalized persistence
diagram; its high-level overview is in \cref{sec:algorithm-summary}.

\subparagraph*{Birth curves.}
Suppose step $(k-\delta_k,l-\delta_l) \to (k,l)$ in the 2-filtration
adds a \emph{negative} simplex $\tsx$. Then the \emph{birth curve} of
$\tsx$ at this step is a periphery $c$ such that if a path $p$
through this step enters the
periphery at step $(i - \delta_i, j - \delta_j) \to (i,j)$, then the simplex
$\ssx$ added along $p$ at this $(i+j)$-th step is paired with $\tsx$ in
the 1-dimensional filtration induced by the path.
\cref{lem:path-invariance} implies that birth curves are well-defined: if $\ssx$
and $\tsx$ are paired along one path through the two steps, then they are paired
along every path through them.

Our algorithm sweeps the 2-filtration and tracks birth curves of negative simplices.
Recall from \cref{sec:preliminaries} that for each birth curve, its
\emph{lower corners} are the minimal grades in its upset;
its \emph{upper corners} are grades $(i,j)$ in the upset,
such that $(i-1,j)$ and $(i,j-1)$ are also in the upset, but $(i-1,j-1)$ is not.

\subparagraph*{Path traversal.}
We start with a path along the left and top edge of the 2-filtration,
$(0,0) \ldots (0,n) \ldots (n,n)$. The filtration that we get from this path is
the same as if we sorted all the simplices by the first coordinate of the
left-most grades in their appearance curves.
We compute persistence $R=DV$ for this filtration. To simplify exposition, for
every unpaired simplex $\ssx$, we add an implicit negative cell $\hat{\ssx}$ at
grade $(n+1,n+1)$, with $R[\hat{\ssx}] = D[\hat{\ssx}] = \ssx$ and
$V[\hat{\ssx}] = \hat{\ssx}$.

We sweep through the paths of the following form,
$(0,0)$ \ldots $(i-1,0)$ \ldots $(i-1,j)$, $(i,j)$ \ldots $(i,n)$ \ldots $(n,n)$,
transitioning one square at a time, by replacing $(i-1,j-1), (i-1,j), (i,j)$ with
$(i-1,j-1), (i,j-1), (i,j)$; see \cref{fig:path-traversal}. As we perform such
elementary steps, we build up the birth curves and report all the non-zero
intervals in the diagram whose upper endpoint is in grade $(i,j)$, i.e.,
all $\Dgm F[\cdot, (i,j)] \neq 0$.

\begin{figure}
    \centering
    \includegraphics{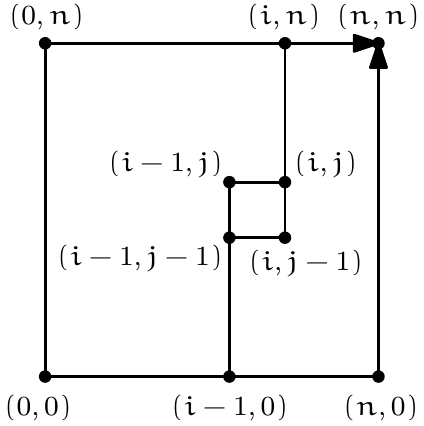}
    \caption{Path traversal starts with the path along the left and top edges of
             the 2-filtration, and through a sequence of elementary steps (one
             of which around $(i,j)$ is shown in the figure) reaches the path
             along the bottom and right edge of the 2-filtration.}
    \label{fig:path-traversal}
\end{figure}

\subparagraph*{Invariant.}
We maintain the following invariant,
necessary to verify the correctness of each step and the running time claim.
We emphasize in \cref{sec:updates} the key parts of the matrix updates that maintain it.

\begin{enumerate}
    \item
        \label{cnd:lowest-one}
        Each negative simplex $\tsx$ along the current path (in the sense of
        \cref{fig:path-traversal}) maintains a birth curve, stored as a set of
        grades that represent its lower corners (by definition, all are below
        the grade of $\tsx$'s appearance along the path). For each lower corner $a$, we
        maintain three chains, $R[\tsx], V[\tsx], V[\ssx]$, such that
        $R[\tsx]$ and $V[\ssx]$ are cycles that appear in the complex $F(a)$,
        but not in any complex $F(a')$ with $a' < a$, and $R[\tsx] = D \cdot V[\tsx]$.
    \item
        Any path that reaches the current grade $(i,j)$
        and then proceeds to grade $(i,n)$ and then $(n,n)$ induces a
        1-filtration. Assembling the columns $R[\tsx],V[\tsx], R[\ssx]=0,
        V[\ssx]$ that are stored at the lower corners below\footnote{Every grade
        on the periphery has a unique lower corner below it.}
        the grades at which
        the path enters the upsets of the birth curves --- ordering all such columns
        with respect to the path --- we get decomposition $R = DV$ that
        satisfies the reduction assumptions ($R$ is reduced, $V$ is invertible
        upper-triangular).
    \item
        \label{cnd:no-spurious-V}
        If for a set of simplices along a path,
        $\ldots \ssx \ldots \tsx \ldots \alpha \ldots \beta \ldots$,
        the pairing is neither nested nor disjoint --- $\ssx$ paired with
        $\alpha$, and $\tsx$ paired with $\beta$ ---
        we ensure that $V[\alpha, \beta] = 0$.
        (This condition is satisfied
        by the original algorithm~\cite{ELZ02}, and
        although the prior work~\cite{vineyards,morozov-phd} does not
        deliberately maintain this property, we explain in
        \cref{sec:transposition-analysis} the necessary extra update, and call
        it out in the text accompanying \cref{fig:pairing-up}.)
\end{enumerate}

\subsection{Updates}
\label{sec:updates}

As we update the path, shifting it around the square $(i,j)$,
suppose only one new simplex, $\ssx$, appears between grades $(i-1,j-1)$ and
$(i,j)$. The analysis of the previous section implies that the diagram can have
a non-zero entry for $\Dgm F[\cdot, (i,j)]$ or $\Dgm F[(i,j), \cdot]$
only if $(i,j)$ is the lower or upper corner of the
appearance curve of $\ssx$. (Otherwise, either $B=0$ or $D=0$ in
\cref{thm:product-formula}.)

If $\ssx$ is positive, we add $(i,j)$ as a lower or upper corner of its pair's birth curve.
If $\ssx$ is negative, then for any
grade $(k,l) < (i,j)$ such that $(k,l)$ is in the upset of the birth curve
$\birth(\ssx)$, but $(k-1,l-1)$ is not, exactly one class is born and dies at
the respective grades along any path
\[
    \ldots, (k-1,l-1), \ldots, (k,l), \ldots, (i-1,j-1), \ldots, (i,j), \ldots
\]
In other words, we are exactly in the setting of \cref{sec:single-pair}.
From the analysis after the proof of \cref{thm:product-formula}, it follows that
if $(i,j)$ is the lower corner of the appearance curve of $\ssx$, then $D=+1$; if
it is the upper corner, then $D=-1$. Similarly, for every lower corner in the
birth curve $\birth(\ssx)$, $B=+1$; for every upper corner, $B=-1$.
Accordingly, we output $\Dgm F[a,b] = B \cdot D$, where $b = (i,j)$ and $a$ is
either the lower or upper corner in the birth curve.

The only remaining possibility is that two simplices, $\ssx$ and $\tsx$, appear
as we go from grade $(i-1,j-1)$ to grade $(i,j)$. Suppose that simplex $\ssx$
appears along the step $(i-1,j-1) \to (i, j-1)$, and simplex $\tsx$ appears
along the step $(i-1,j-1) \to (i-1,j)$.
In the remainder of this section, we analyze all possible scenarios involving
such $\ssx$ and $\tsx$.


\subparagraph*{$\ssx$ and $\tsx$ are both positive.}
Suppose $\ssx$ is paired with $\alpha$, and $\tsx$ is paired with $\beta$.
If $\beta$ comes first, then by \cref{lem:nested-disjoint} pairing of $\ssx$ and
$\tsx$ cannot switch between the two paths.
It is, however, possible that the columns of $R[\alpha]$,
$V[\alpha]$, $V[\ssx]$, and $R[\beta]$, $V[\beta]$, $V[\tsx]$
need to be updated.
Such an update can be performed in linear time~\cite{vineyards}; see \cref{sec:transposition-analysis}.

We note that the update of columns $V[\ssx]$ and $V[\tsx]$ is crucial in this
case.  It ensures that if $V[\ssx]$ contains $\tsx$ before the
transposition, then it doesn't after the transposition. This ensures correctness of
\cref{cnd:lowest-one} in the invariant:
column $V[\ssx]$ contains only simplices present at the lower corner
that stores it --- a corner yet to be reached in this case.

The only way the pairing of $\ssx$ and $\tsx$ can switch is
if $\alpha$ comes before $\beta$, as in \cref{fig:pairing-up}.
The two paths induce the following simplex orders: $\ldots \tsx \ssx \ldots \alpha \ldots
\beta$ (before) and $\ldots \ssx \tsx \ldots \alpha \ldots \beta$ (after).
\begin{figure}
    \includegraphics[width=\textwidth]{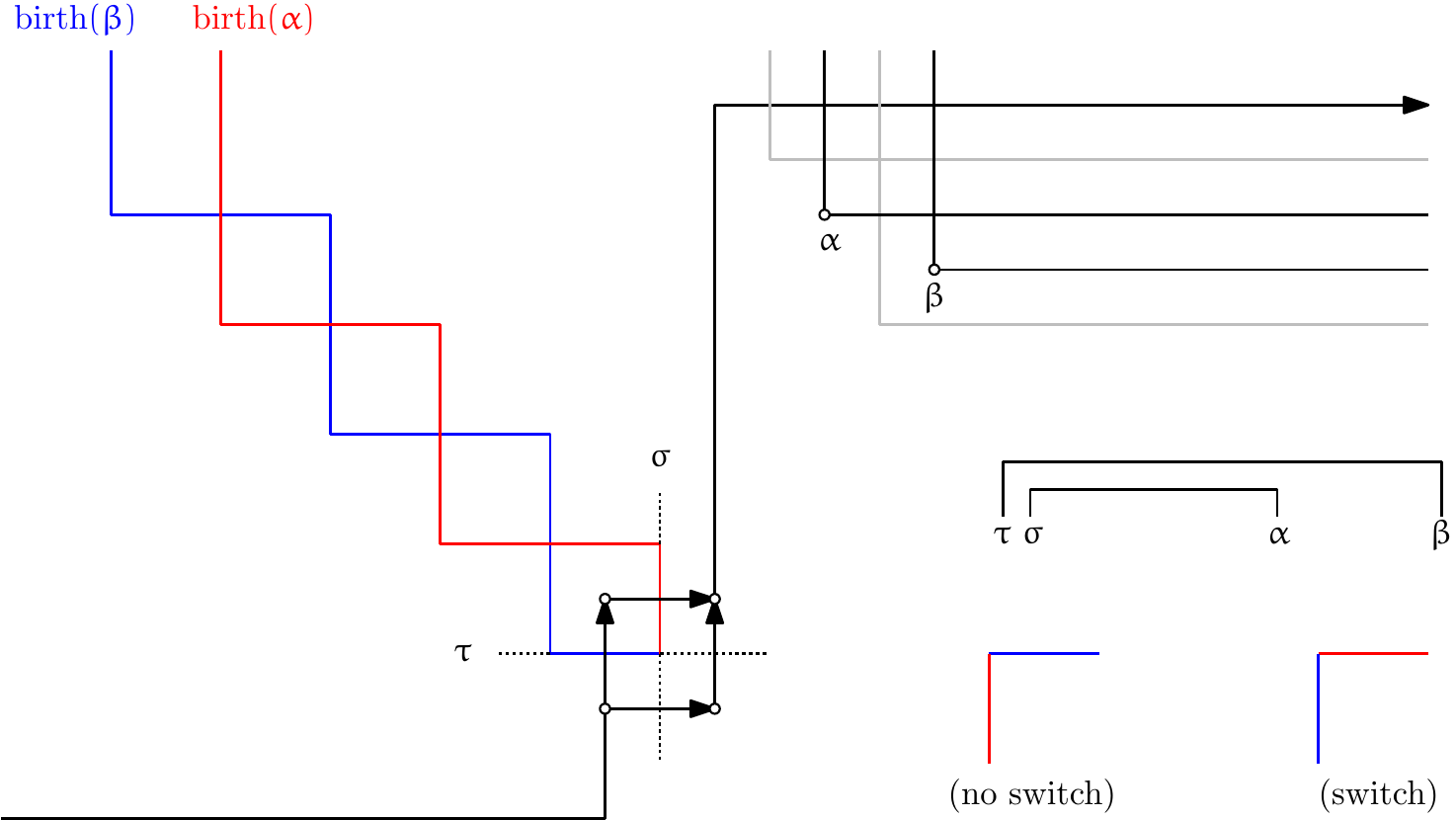}
    \caption{Positive simplices $\ssx$ and $\tsx$. The birth curve $\birth(\alpha)$ is shown in red; the
             birth curve $\birth(\beta)$, in blue.}
    \label{fig:pairing-up}
\end{figure}
We can determine in constant time whether the pairing of $\ssx$ and $\tsx$
switches between the two paths. Depending on the answer, we update the birth
curves of $\alpha$ and $\beta$ in one of the two ways, shown in
\cref{fig:pairing-up}. We note that if the pairing switches,
$(i,j)$ becomes a lower corner of the birth
curve $\birth(\alpha)$, and the upper corner of the birth curve $\birth(\beta)$.
We store the updated columns $R[\alpha]$, $V[\alpha]$, $R[\tsx]$ with the new
lower corner of $\birth(\alpha)$.

If the pairing does not switch, and thus goes from nested to neither nested nor
disjoint, it is crucial to update the column $V[\beta]$, stored at the lower
corner of $\birth(\beta)$ defined by $\tsx$, to ensure $V[\alpha,\beta] = 0$
(necessary for \cref{cnd:no-spurious-V} in the invariant).
An example of such an update is spelled out in Case 1 in
\cref{sec:transposition-analysis}; it takes linear time.

\subparagraph*{$\ssx$ is positive, $\tsx$ is negative.}
This scenario is illustrated in \cref{fig:pairing-down-up}.
Consider any path that reaches grade $(i-1,j-1)$ and then proceeds like the
first path in \cref{fig:path-traversal}. Suppose it induces an ordering of
simplices
$\ldots \beta \ldots \tsx \ssx \ldots \alpha$, where $\beta$ is paired with
$\tsx$ and $\ssx$ is paired with $\alpha$.
We consider the path that differs by the transposition of $\ssx$ and $\tsx$,
$\ldots \beta \ldots \ssx \tsx \ldots \alpha$.

If the pairing of $\ssx$ and $\tsx$ switches for one such path, it switches for
all such paths.
%
\begin{figure}
    \includegraphics[width=\textwidth]{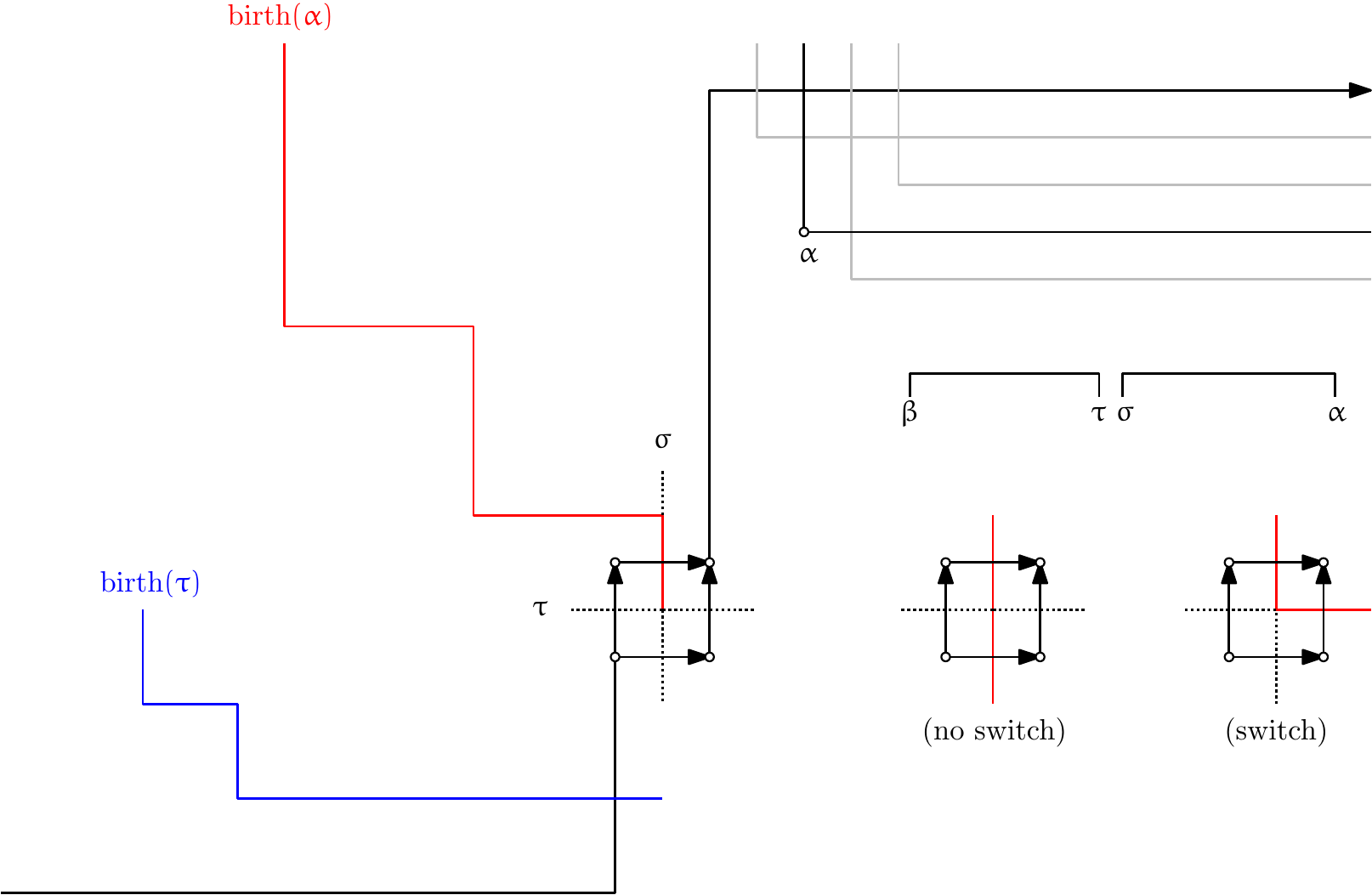}
    \caption{Positive $\ssx$ and negative $\tsx$. The birth curve $\birth(\alpha)$, in
             red; $\birth(\tsx)$, in blue.}
    \label{fig:pairing-down-up}
\end{figure}
We can determine in constant time whether this happens. If it
doesn't, there is nothing to report, and there is no need to
update columns $R[\alpha], V[\ssx]$ because $\tsx$ cannot appear
in them (this follows from \cite{vineyards}; see
\cref{sec:transposition-analysis}).
%
If the pairing does switch, we update the birth curve $\birth(\alpha)$ as shown
in \cref{fig:pairing-down-up}; grade $(i,j)$ becomes the lower corner of the birth curve.
$\ssx$ takes over the birth curve of $\tsx$.
For every grade entering the upset of the birth curve, we are in the setting of
\cref{sec:single-pair}. For \cref{thm:product-formula}, $D=-1$ because whichever
simplex appears first going from $(i-1,j-1)$ to $(i,j)$ kills the respective
class. For every lower corner in the birth curve, $B=+1$; for every upper
corner, $B=-1$.  Accordingly, we output $\Dgm(a,b) = -1$ for all grades $a$ that
are lower corners in $\birth(\ssx)$ and $b = (i,j)$;
and $\Dgm(a,b) = +1$ for all $a$ that are the upper corners in $\birth(\ssx)$.

Let us dwell for a moment on the updates to the columns of $V$, when the pairing
does switch.  As explained in \cref{sec:transposition-analysis} (Case 3), the
necessary update, encoded in matrix $X$, subtracts a multiple $\lambda$ of the
column $V[\tsx]$ from $V[\ssx]$ before the transposition (to produce matrix $VX$ in
\cref{sec:transposition-analysis}),
and then adds $(V[\ssx] - \lambda V[\tsx])$ to $\lambda V[\tsx]$
(after the transposition, to produce matrix $V' = PVXPZ$). In other words,
$V'[\tsx] = V[\ssx]$ and $V'[\ssx] = V[\ssx] - \lambda V[\tsx]$. The former equality
means that the birth curve $\birth(\alpha)$ doesn't require any updates. But the
latter equality means that for every lower corner of the birth curve
$\birth(\tsx)$, we need to add $V[\ssx]$ to $-\lambda$ multiplied by $V[\tsx]$ stored at that corner.
This update takes linear time per corner, but it's required only if the pairing
switches, in which case every corner contributes to a non-zero interval in the
generalized persistence diagram. We charge each such linear-time update to the
output, i.e., to the $O(Cm)$ term in the running time.

In summary, we can detect whether a switch in the pairing occurs --- and if it
doesn't, perform the necessary updates --- in linear time.
If the switch does occur, we update each step in the birth curve, but each such
update corresponds to an interval in the output.

\subparagraph*{$\ssx$ is negative, $\tsx$ is positive.}
In this case, the pairing is neither nested, nor disjoint, so by
\cref{lem:nested-disjoint} it cannot switch. But the birth curve $\birth(\ssx)$
may contain a lower corner at grade
$(\cdot,j)$, i.e., there exists a path along which $\tsx$ and $\ssx$ are
paired. We remove this corner.

\subparagraph*{$\ssx$ and $\tsx$ are negative.}
This scenario is illustrated in \cref{fig:pairing-down}. Because both simplices
are negative, each one has its own birth curve, $\birth(\ssx)$ and
$\birth(\tsx)$.
\begin{figure}
    \includegraphics[width=\textwidth]{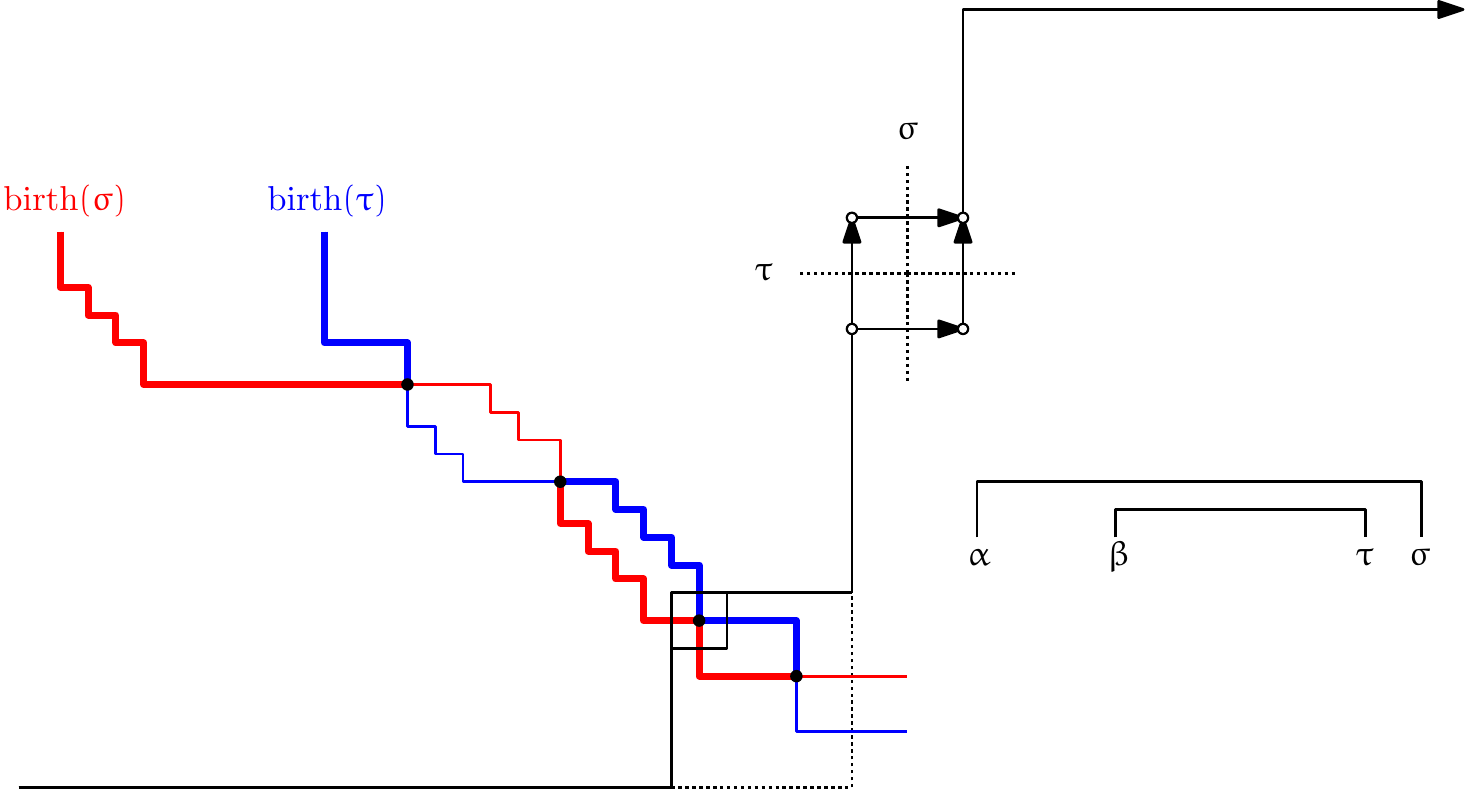}
    \caption{Negative $\ssx$ and $\tsx$.  The birth curve $\birth(\ssx)$, in
             red; the birth curve $\birth(\tsx)$, in blue.}
    \label{fig:pairing-down}
\end{figure}
We can split the birth curves into two types of segments: those where the birth
curve of $\ssx$ lies below that of $\tsx$, and vice versa. The pairing of $\tsx$
and $\ssx$ can switch only in filtrations induced by the paths through the former
(highlighted in bold in \cref{fig:pairing-down}). This follows
from \cref{lem:nested-disjoint}: only for such paths is the pairing of the two
simplices nested. (Moreover, for any path through the latter type of segment,
$V[\tsx,\ssx] = 0$ --- thanks to \cref{cnd:no-spurious-V} in the invariant
--- so no update is necessary for these segments.)

It follows from the stability of 1-dimensional persistence that if the pairing
of $\ssx$ and $\tsx$ switches for one path through the segment, it switches for
all paths through the segment. Accordingly, it suffices to only check the paths around
the grades, where the two birth curves intersect. We can locate all such
intersections in linear time.

There are four paths around the two corners:
\begin{center}
\begin{tabular}{ccc}
    $\ldots \alpha \beta  \ldots \tsx \ssx$  & \qquad & $\ldots \alpha \beta  \ldots \ssx \tsx$ \\
    $\ldots \beta  \alpha \ldots \tsx \ssx$  & \qquad & $\ldots \beta  \alpha \ldots \ssx \tsx$
\end{tabular}
\end{center}
We consider all combinations. In all figures, red signifies the pair of $\ssx$
and blue, the pair of~$\tsx$.

\subparagraph{Case A.}
Suppose the pairing of the first two paths (in which $\tsx$ comes before $\ssx$)
is as shown in the figure. Then there are two possibilities: either the pairing
switches when we transpose either pair of simplices in the top-right path,
$\ldots \alpha \beta \ldots \ssx \tsx \ldots$, or it doesn't. We note that if it
switches for one of the transpositions, it is forced to switch for both of them.
If the pairing doesn't switch, it means that $V[\tsx,\ssx] = 0$ for all columns
$V[\ssx]$ stored in the birth curve $\birth(\ssx)$, meaning there is nothing to
update.

\vspace{1ex}
\begin{center}
\includegraphics[page=2,width=\textwidth]{figures/pairing-down-cases}
\end{center}
If there is a switch in the pairing, the birth curves are updated (as shown on
the right of the figure) by swapping the respective segments.
The update of each lower corner along the segment
takes linear time, but each such corner also produces an interval in the
persistence diagram. So we charge the update to the output.
(We note that when updating the columns at the lower corners $a \in
\birth(\tsx)$ of the upper birth curve, we have to choose a corner $b \in
\birth(\ssx)$ of the lower birth curve that lies below $a$ to perform an update.
There can be multiple such choices, but any one of them works.)

The intervals reported in this case are
$\Dgm[a,b] = +1$ for $a$ in the lower corners of
$\birth(\tsx)$ or the upper corners of $\birth(\ssx)$
(the corners are restricted to the appropriate segments, and the birth curve
taken after the switch),
and $\Dgm[a,b] = -1$ for $a$ in the upper corners of
$\birth(\tsx)$ or the lower corners of $\birth(\ssx)$, with $b = (i,j)$.
(We note that for the lower and upper corners in the interiors of the segments,
i.e., away from their intersection,
these results follow from \cref{thm:product-formula} in \cref{sec:single-pair}. At the intersection point,
we are in the setting of \cref{sec:two-pairs}, and the result follows from the
second case in that section.)

\subparagraph{Case B.}
This case is symmetric to Case A. That case occurs at the bottom of a
segment; this case occurs at the top.

\vspace{1ex}
\begin{center}
\includegraphics[page=3,width=\textwidth]{figures/pairing-down-cases}
\end{center}

\subparagraph{Case C.}
The pairing shown in the figure is impossible since it implies that the pairing
switches for a transposition of pairs that are neither nested, nor disjoint,
violating \cref{lem:nested-disjoint}. (We note that there is no
contradiction with the previous figures, since there, after the transposition,
$\ssx$ comes before $\tsx$.)

\vspace{1ex}
\begin{center}
\includegraphics[page=4,width=\textwidth]{figures/pairing-down-cases}
\end{center}

\subparagraph{Case D.}
Suppose the pairing of the first two paths (in which $\tsx$ comes before $\ssx$)
is as shown in the next three figures. There are three possibilities: either the
pairing switches after the transposition of $\ssx$ and $\tsx$ in the second
path, but not the first; or it switches in the first path, but not the second;
or it switches in both. It's impossible for the pairing to remain the same
along both paths without violating \cref{lem:nested-disjoint}.

In the first two cases, the pairing switches for one of the two segments of the
birth curves that end at the intersection point in grade $b = (i,j)$; in the
third case, it switches for both.
The segments for which the pairing does change require updates to the
columns of the matrices $R$ and $V$ stored at their low corners, but each such
corner also results in an interval in the diagram, and we charge the update to
the output.
For the segments where the pairing doesn't switch, we have $V[\tsx,\ssx] = 0$
for all columns $V[\ssx]$ stored in the birth curve $\birth(\ssx)$, meaning
there is nothing to update.

The intervals $(a,b)$, for $a$ in the lower or upper corners of the birth curves, get
$+1$ or $-1$ as in Cases A and B. Specifically, taking the
birth curves before the pairing update, we output
$\Dgm[a,b] = +1$ for $a$ in the lower corners of $\birth(\ssx)$ and the upper corners $\birth(\tsx)$;
$\Dgm[a,b] = -1$ for $a$ in the upper corners of $\birth(\ssx)$ and the lower corners $\birth(\tsx)$.
(This again follows from \cref{thm:product-formula} in \cref{sec:single-pair}.)

The exception is when $a$ is the grade of the intersection of the two curves,
i.e.,
the grade depicted in the figures.
Here, the analysis in \cref{sec:two-pairs}
applies: we get $\Dgm [a,b] = -1$ in
the first two cases, and $\Dgm [a,b] = -2$ in the third case.

\vspace{1ex}
\begin{center}
\includegraphics[page=5,width=\textwidth]{figures/pairing-down-cases}
\end{center}
\vspace{1ex}
\begin{center}
\includegraphics[page=6,width=\textwidth]{figures/pairing-down-cases}
\end{center}
\vspace{1ex}
\begin{center}
\includegraphics[page=7,width=\textwidth]{figures/pairing-down-cases}
\end{center}

After we perform all the updates, if the birth curve $\birth(\ssx)$ has a lower
corner at grade $(\cdot, j)$, we remove it. (This cannot happen
in the previous case of disjoint pairing.)



\subparagraph*{Infinite intervals.}
After the traversal, we output the ``infinite'' intervals
$\Dgm [a, (n,n)] = +1$ and $\Dgm [b, (n,n)] = -1$ for the lower corners $a$ and
the upper corners $b$ in the birth curves
of the implicit cells $\hat{\ssx}$.

\subsection{Analysis}
\label{sec:analysis}

After the initial $O(m^3)$ persistence computation, the algorithm takes $O(n^2)$ steps.
Each step requires an $O(m)$ update, plus an update of the birth curves that we
charged to the output: $O(m)$ time for each one of the $C$ intervals in the
output. The total running time is in $O(n^2m + Cm)$. It is immediate from
the algorithm that the size of the output $C$ is in $O(n^3)$, making the whole
algorithm no worse than $O(n^4)$ brute-force approach. On the other hand,
$C$ can be as low as $n$: for example, if the entire 2-filtration is totally
nested, i.e., if the grades of every pair of simplices are comparable in the
poset.

\section{Conclusion}

Generalized persistence diagrams are a promising direction for TDA research.
They generalize all the desirable properties of 1-parameter persistence used in
applications. Their structure offers the possibility of following the
1-parameter blueprint, including straightforward adaptation of the newer
methodologies, such as vectorizing persistence diagrams for use in machine
learning~\cite{PMK22,perslay} or modifying the input data by back-propagating gradients
through persistence diagrams~\cite{LOT22,NM22}.

Until now, the crucial missing piece was efficient computation. We hope that the
output-sensitive algorithm presented in this paper will pave the way for using
generalized persistence diagrams in applications.

\bibliographystyle{hplain}
\bibliography{references}

\clearpage
\appendix
\section{Transposition Analysis}
\label{sec:transposition-analysis}

The updates to 1-dimensional persistence when two consecutive simplices, $\tsx$
and $\ssx$, transpose are studied in \cite{vineyards,morozov-phd}.  Given a
boundary matrix $D$ of a filtration, together with its decomposition $R = DV$
into a reduced matrix $R$ and an invertible upper-triangular matrix $V$, let $P$
denote the permutation matrix that transposes two adjacent columns (if
multiplied on the right) or rows (if multiplied on the left) that correspond to
the simplices $\tsx$ and $\ssx$. We are interested in finding the decomposition
$R' = (PDP) V'$, where $R'$ is reduced and $V'$ is invertible upper-triangular.

\begin{figure}
    \centering
    \includegraphics[width=\textwidth]{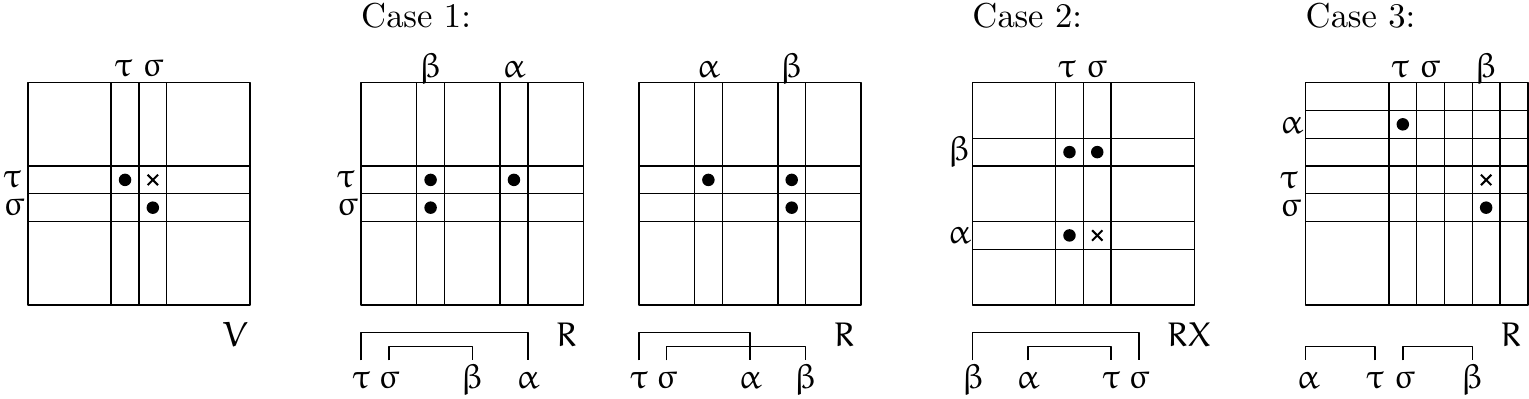}
    \caption{Updates of the matrices $R$ and $V$ after the transposition of two
             simplices.}
    \label{fig:RDV}
\end{figure}

\subparagraph*{Matrix $V$ update.}
The only way in which $V' = PVP$ can fail to satisfy this condition is if
$V[\tsx,\ssx] \neq 0$. If this is the case, we denote by $X$ the matrix that
subtracts an appropriate multiple $\lambda$ of column $V[\cdot, \tsx]$ from
column $V[\cdot, \ssx]$ to make $V[\tsx,\ssx] = 0$.
In this case, matrix $PVXP$ is invertible upper-triangular.
We abuse the notation, and let $X$ denote either this update matrix, or
identity if no update was required.

\subparagraph*{Case 1.}
Suppose both simplices $\tsx$ and $\ssx$ are positive, paired with
simplices $\alpha$ and $\beta$ respectively. (The case where either one of the
simplices is unpaired is analogous.) In this case, it's possible that columns
$R[\cdot,\alpha]$ and $R[\cdot,\beta]$ are such that $PRXP$ is not reduced.
This happens when $R[\tsx,\beta] \neq 0$; see \cref{fig:RDV}.
If this happens, we can subtract an appropriate multiple of the first of these
two columns from the second, to
ensure that the new matrix is reduced. Denoting this subtraction with matrix $Y$,
we get decomposition $(PRXPY) = (PDP) (PVXPY)$.

If the pairing was nested before the transposition, but did not switch after the
transposition (and thus became neither nested, nor disjoint), i.e., if $PRXP$ is
already reduced, it is possible for the entry $V[\beta,\alpha] \neq 0$. The
prior work~\cite{vineyards,morozov-phd} does not pay any special attention to
this case, but, in order to satisfy \cref{cnd:no-spurious-V} in the invariant in
\cref{sec:algorithm}, we need to subtract an appropriate multiple of the
column $V[\cdot,\beta]$ from the column $V[\cdot,\alpha]$. Denoting this update with matrix
$Y'$, we get a decomposition $(PRXPY') = (PDP) (PVXPY')$. In this case, matrix
$PRXPY'$ is necessarily reduced.
The pairs of any $\alpha'$ that were added as non-zero entries
$V'[\alpha',\alpha]$ are necessarily nested in the pair $\ssx$-$\alpha$.

\subparagraph*{Case 2.}
Suppose both simplices $\tsx$ and $\ssx$ are negative, paired with
simplices $\alpha$ and $\beta$ respectively. In this case, if $\alpha$ comes
after $\beta$, then the columns $(PRXP)[\cdot, \tsx]$ and $(PRXP)[\cdot, \ssx]$
may not be reduced because of the update caused by matrix $X$; see
\cref{fig:RDV}. In this case, we can apply matrix $Z$ after the
transposition. This matrix replaces the later column (of simplex $\tsx$ after
the transposition) by multiplying it by $\lambda$ and adding an earlier column.
In other words, we get
\begin{align*}
    (PRXPZ)[\cdot,\tsx] &= (PRXP)[\cdot,\ssx] + \lambda (PRXP)[\cdot,\tsx] \\
                        &= \left((PRP)[\cdot,\ssx] - \lambda (PRP)[\cdot,\tsx]\right) + \lambda (PRP)[\cdot,\tsx]       \\
                        &= (PRP)[\cdot,\ssx].
\end{align*}
(The same analysis applies to $(PVXPZ)[\cdot,\tsx]$.)
In other words, the second of the two columns in matrices $R$ and $V$ do not change.
The resulting matrix is reduced and the decomposition, $(PRXPZ) = (PDP) (PVXPZ)$,
satisfies the two conditions.

If the pairing went from nested to neither nested nor disjoint (i.e., it did not switch), the
original update $X$ to matrix $V$ ensures that \cref{cnd:no-spurious-V} in the
invariant in \cref{sec:algorithm} is satisfied.

\subparagraph*{Case 3.}
Suppose simplex $\tsx$ is negative, while simplex $\ssx$ is positive; again the
two are paired with $\alpha$ and $\beta$ respectively. If matrix $V$ required an
update, then because the column $R[\cdot,\ssx] = 0$,
the columns $\tsx$ and $\ssx$ in matrix $(PRXP)$ are the same, up to the factor
of $-\lambda$,
requiring a further reduction by an application of matrix $Z$ from Case 2.
We get decomposition, $(PRXPZ) = (PDP) (PVXPZ)$.
We note that it is not immediately obvious, but true that
$R[\tsx,\beta] \neq 0$ iff $V[\tsx,\ssx] \neq 0$.




\section{Algorithm Summary}
\label{sec:algorithm-summary}

\begin{algorithm}
    \DontPrintSemicolon
    \BlankLine
    compute 1-parameter persistence $R = DV$, simplices sorted by the lowest
    $x$-coordinate, $x(\ssx)$, in their appearance curve \;
    \ForEach{positive $\ssx$}{
        \If{$\ssx$ is paired with $\tsx$}{
            $\birth(\tsx) = ((R[\tsx],V[\tsx],V[\ssx], (x(\ssx),n)))$ \;
        }
        \ElseIf{$\ssx$ is unpaired}{
            pair $\ssx$ with $\hat{\ssx}$ implicitly added at grade $(n+1,n+1)$
                \nllabel{line:unpaired} \;
            $\birth(\hat{\ssx}) = ((R[\hat{\ssx}] = \ssx, V[\hat{\ssx}] = \hat{\ssx},V[\ssx], (x(\ssx),n)))$ \;
        }
    }
    \caption{High-level overview}
    \label{alg:sparsify}
\end{algorithm}

\setlength{\interspacetitleruled}{0pt}%
\setlength{\algotitleheightrule}{0pt}%
\begin{algorithm}
    \DontPrintSemicolon
    \setcounter{AlgoLine}{7}
    \For{$i=1$ \KwTo $n$}{
        \For{$j=n$ \KwTo $1$}{
            \If{$(i,j)$ is the lower-corner in the appearance curve of some $\ssx$}{
                \If{$\ssx$ is positive, paired with $\tsx$}{
                    set the $y$-coordinate of the last corner in $\birth(\tsx)$ to $j$\;
                }
                \ElseIf{$\ssx$ is negative}{
                    output $(+1, a, (i,j))$ for $a \in l(\birth(\ssx))$ \;
                    output $(-1, a, (i,j))$ for $a \in u(\birth(\ssx))$ \;
                }
            }
            \ElseIf{$(i,j)$ is the upper-corner in the appearance curve of some $\ssx$}{
                \If{$\ssx$ is positive, paired with $\tsx$}{
                    make a copy of the last corner in $\birth(\tsx)$, setting its
                    grade to $(i,j)$ (component $j$ will be updated later) \;
                }
                \ElseIf{$\ssx$ is negative}{
                    output $(-1, a, (i,j))$ for $a \in l(\birth(\ssx))$ \;
                    output $(+1, a, (i,j))$ for $a \in u(\birth(\ssx))$ \;
                }
            }
            \ElseIf{$(i,j)$ is the grade where $\ssx$ and $\tsx$ appear together for the first time}{
                \If{$\tsx$ and $\ssx$ are positive}{
                    let $\alpha$ be the pair of $\ssx$ and $\beta$, the pair of $\tsx$ \;
                    update the birth curve and $R[\alpha]$, $V[\alpha]$,
                        $V[\ssx]$, $R[\beta]$, $V[\beta]$, $V[\tsx]$ as described
                        in the text accompanying \cref{fig:pairing-up} \;
                }
                \ElseIf{$\tsx$ is negative, $\ssx$ is positive}{
                    let $\alpha$ be the pair of $\ssx$ \;
                    determine if the pairing switches (text accompanying \cref{fig:pairing-down-up}), extend $\birth(\alpha)$ accordingly \;
                    \If{the pairing switches}{
                        $\birth(\ssx)$ takes over $\birth(\tsx)$ \;
                        update the columns stored at the corners of $\birth(\ssx)$ \;
                        output $(-1, a, (i,j))$ for $a \in l(\birth(\ssx))$ \;
                        output $(+1, a, (i,j))$ for $a \in u(\birth(\ssx))$ \;
                    }
                }
                \ElseIf{$\tsx$ and $\ssx$ are negative}{
                    identify distinct segments of the birth curves $\birth(\tsx)$, $\birth(\ssx)$ \;
                    \ForEach{segment}{
                        determine if the pairing switches \;
                        update and output as described in the text accompanying \cref{fig:pairing-down}
                    }
                    \lIf{the first corner in $\birth(\ssx)$ is $(\cdot,j)$}{remove it}
                }
                \ElseIf{$\tsx$ is positive, $\ssx$ is negative}{
                    \lIf{the first corner in $\birth(\ssx)$ is $(\cdot,j)$}{remove it}
                }
            }
        }
    }
    \ForEach{$\hat{\ssx}$ added in Line \ref{line:unpaired}}{
        output $(+1, a, (n,n))$ for $a \in l(\birth(\hat{\ssx}))$ \;
        output $(-1, a, (n,n))$ for $a \in u(\birth(\hat{\ssx}))$ \;
    }
\end{algorithm}

\end{document}